\documentclass[a4paper,12pt]{article}

\usepackage{amssymb,amsmath,amsfonts,amsthm}

\usepackage{a4wide}

\usepackage{authblk}

\newtheorem{definition}{Definition}
\newtheorem{thm}{Theorem}
\newtheorem{lem}{Lemma}
\newtheorem{prop}{Proposition}
\newtheorem*{remark}{Remark}

\title{Large Time Behaviour and Convergence Rate for Non Demolition Quantum Trajectories}

\author[$\spadesuit$]{Tristan Benoist\thanks{tristan.benoist@ens.fr}}
\author[$\clubsuit$]{Cl\'ement Pellegrini\thanks{clement.pellegrini@math.univ-toulouse.fr}}

\affil[$\spadesuit$]{\small Laboratoire de Physique Th\'eorique de l'ENS,}
 \affil[$$]{CNRS \& \'Ecole Normale Sup\'erieure de Paris.}
\affil[$$]{24 rue Lhomond 75005 Paris, France}
\affil[$$]{}

\affil[$\clubsuit$]{\small Institut de Math\'ematiques de Toulouse}
\affil[$$]{Equipe de Statistique et de Probabilit\'e}
\affil[$$]{Universit\'e Paul Sabatier}
\affil[$$]{31062 Toulouse Cedex 9, France}

\begin{document}

 \maketitle

 \begin{abstract}
 A quantum system $\mathcal S$ undergoing continuous time measurement is usually described by a jump-diffusion stochastic differential equation. Such an equation is called a stochastic master equation and its solution is called a quantum trajectory. This solution describes actually the evolution of the state of $\mathcal S$. In the context of Quantum Non Demolition measurement, we investigate the large time behavior of this solution. It is rigorously shown that, for large time, this solution behaves as if a direct Von Neumann measurement  has been performed at time 0. In particular the solution converges to a random pure state which is related to the wave packet reduction postulate. Using theory of Girsanov transformation, we determine precisely the exponential rate of convergence towards this random state. The important problem of state estimation (used in experiment) is also investigated.
 \end{abstract}

 \section*{Introduction}
 In quantum optics, indirect measurements are often used \cite{Har1,haroche}. Usually a system is probed by light beams (direct photodetection, homodyne and heterodyne detection schemes) or conversely, atoms probe a photon field trapped in a cavity. Such experiments are promising towards the manipulation of quantum states \cite{haroche,sayrin,W1,Rou4,Rou5,Mi1}. They are designed to extract information from a quantum system on site and without destroying it. The idea is to avoid direct interaction of the quantum system with a macroscopic apparatus (photo detector, screen ...). Instead  the physical setup is the following: a quantum system $\mathcal S$ (from which we want to extract information) is put in interaction with an auxiliary quantum system $\mathcal E$. After interaction, a measurement on $\mathcal E$ is performed. Due to the entanglement between $\mathcal S$ and $\mathcal E$, the result of the measurement 
holds some information on $\mathcal S$. Conditionally to this result one can compute the 
evolution of $\mathcal S$. 
 
 One of the best example of such setups is Serge Haroche's group experiment at LKB\cite{guerlin}. They have successfully implemented a scheme of repeated interactions and measurements which allows to measure the number of photons in a cavity (without destroying the photons). The principle consist in putting the system $\mathcal S$ (the cavity photon field) in contact with a sequence of auxiliary systems (Rydberg atoms) $\mathcal E_k$ which interact one after the other with $\mathcal S$. After each interaction a measurement on the atom $\mathcal E_k$ which has just finish to interact is performed. Such a procedure, called \textit{repeated quantum indirect measurements}, allows to monitor the system $\mathcal S$ and to have an estimation of the number of photons inside the cavity.
 
 A particular feature in the Serge Haroche's group experiment is that only Quantum Non Demolition (QND) measurement are performed. Such a scheme is at the cornerstone of the mathematical study of the long time behavior of $\mathcal S$. In \cite{Bauer1,Bauer2}, the authors show that the state of the system $\mathcal S$ converges when the number of interactions tends to infinity. More precisely they show a convergence which is compatible with the wave function collapse postulate. In particular it is shown that the state of $\mathcal S$ behaves in infinite time as if a direct Von Neumann measurement on $\mathcal S$ would have been performed at time $0$. Essentially these results concern discrete time model where the time of interaction $\tau$ between $\mathcal S$ and a piece $\mathcal E_k$ is fixed. They apply to general nondemolition measurement scheme of which Serge Haroche's group experiment is an example (see \cite{rev_devoret} and references therein).
 
 When the time of interaction goes to zero, this yields to continuous time models. In \cite{FLA} it has been shown that quantum repeated interactions model are a powerful approximation of the so called \textit{Quantum Langevin equation}. In \cite{Pe1,Pe2,Pe3}, it is shown that the continuous time approximation ($\tau$ goes to zero) of repeated quantum indirect measurements lead to jump-diffusion stochastic differential equations (see also \cite{Bauer3}). Such equations are namely the equations which describe the evolution of a quantum system undergoing indirect continuous measurements \cite{Bar1,Bar2,Bar3,Bar4,Bar5,Bar6,Gisin1,francesco2,carm0,W1}. They are called \textit{stochastic master equations} and their solutions \textit{quantum trajectories}. 
 
In this article, we focus on the stochastic master equations describing general continuous time quantum nondemolition measurement. Our main purpose is to describe the long time behavior of the state of $\mathcal S$ when the time $t$ goes to infinity. In particular if $(\rho(t))$\footnote{The process $(\rho(t))$ is actually the quantum trajectory describing the evolution of the state of $\mathcal S$ which undergoes indirect continuous measurement} describes the stochastic evolution of $\mathcal S$ undergoing indirect QND measurement, we show that  $(\rho(t))$ converges to a pure state $\vert \Upsilon\rangle\langle\Upsilon\vert$. This convergence is obtained by studying in detail the quantities $(q_\alpha(t))$ defined by

 $$q_\alpha(t)=\textrm{Tr}[\rho(t)\vert \alpha\rangle\langle\alpha\vert], \alpha\in\mathcal P,  t\geq0$$
 where $\mathcal P$ is a preferred basis of the Hilbert space describing $\mathcal S$. In particular the quantity $q_\alpha(t)$ gives the probability for $\mathcal S$ to be in the state $\vert\alpha\rangle$ if a direct measurement on $\mathcal S$ would have ben performed at time $t$. The elements $\vert\alpha\rangle$ are often called pointer states. We show that $\Upsilon$ is a random variable on the set of pointer states. By studying the distribution of $\Upsilon$ we are able to connect the convergence towards $\vert \Upsilon\rangle\langle\Upsilon\vert$ with the wave function collapse postulate at time $0$. From this convergence we study the behavior of the system conditioned on the limit state $\Upsilon$ being $\gamma$. This conditioning corresponds to a particular martingale change of measure. Using a standard Girsanov transformation, we show rigorously that the convergence towards $\vert \Upsilon\rangle\langle\Upsilon\vert$ is exponentially fast and we give its explicit rate. The problem of 
estimation is also investigated when we are in the situation where the initial state $\rho(0)$ is unknown. In this context, we only have access to the result of the measurements. Since $\rho(0)$ is assumed to be unknown, this is 
not sufficient for describing totally $(\rho(t))$. In particular we are only able to describe the evolution of an estimate state of $(\rho(t))$ denoted by $(\tilde \rho(t))$. We show that this estimation is efficient since $(\tilde \rho(t))$ converges in long time to the same limit $\vert \Upsilon\rangle\langle\Upsilon\vert$. Such a property often refers to the notion of stability of quantum filter \cite{Van1,Rou1,Rou2,Rou3,Mi1}.
 \bigskip
 
 The article is structured as follows. In Section 1, we introduce the stochastic models describing the generic stochastic master equations. Next we present the particular case of nondemolition stochastic master equations. This allows us to define the processes $(q_\alpha(t))$. We then study the properties of these processes and show that they are bounded martingales. Section 2 is devoted to the main convergence theorem. From the martingale and boundedness property of $(q_\alpha(t))$, we conclude that these processes converge almost surely. This allows us to present the main convergence result and to define precisely the random variable $\Upsilon$. Next, using appropriate Girsanov change of measure, we show that this convergence is exponentially fast. Finally we investigate the problem of estimation.

 \section{Non destructive quantum trajectories}
  \subsection{System state evolution}
  
  This section is devoted to present the continuous time stochastic processes which describe quantum trajectories. As announced these stochastic processes are solutions of particular type of jump-diffusion stochastic differential equations. 
  
 Before presenting the SDEs, let us introduce some notations. The quantum system is represented by a finite dimensional Hilbert space denoted by $\mathcal{H}$. We denote the set of density matrices by $\mathcal S (\mathcal H)=\{\rho\in\mathcal B(\mathcal H),\rho\geq 0,\textrm{Tr}[\rho]=1\}$. A density matrix represents a general system mixed state. A system in a pure state $\vert \phi \rangle \in \mathcal H$ corresponds to a special case where the density matrix is the projector onto $\vert \phi\rangle$. In this situation the corresponding density matrix is $\rho_{\vert \phi \rangle}=\vert \phi\rangle \langle\phi \vert$. In the rest of the article, if not specified, the term state refers to a density matrix.
 
 Let us consider a family $C_i,i=0,\ldots,n$ of operators in $\mathcal B (\mathcal H)$ and let $H\in\mathcal B(\mathcal H)$ such that $H=H^*$ i.e. $H$ is a Hermitian operator. On $\mathcal S(\mathcal H)$, we introduce the following functions:
   \begin{equation}\label{eq:def_functions}
   \begin{split}
    L(\rho)&=-i[H,\rho] + \sum_{i=0}^n\left( C_i\rho C_i^* - \frac12 \big(C_i^*C_i \rho + \rho C_i^*C_i\big)\right)\\
   J_i(\rho)&=C_i\rho C_i^*,\,\, i=0,\ldots,n\\
   v_i(\rho)&=\textrm{Tr}[J_i(\rho)],\,\, i=0,\ldots,n\\
  H_i(\rho)&=C_i\rho + \rho C_i^* - {\rm Tr}[(C_i + C_i^*)\rho]\rho,\,\, i=0,\ldots,n,
   \end{split}
  \end{equation}
 for all states $\rho\in\mathcal S(\mathcal H)$.
 
 Let $(\Omega,\mathcal F,(\mathcal F_t),\mathbb P)$ be a filtered probability space with usual conditions. Let $(W_j(t)),j=0,\ldots,p$ be standard independent Wiener processes and let $(N_j(dx, dt)),j=p+1,\ldots,n$ be independent adapted Poisson point processes of intensity $dxdt$; the $N_j$'s are independent of the Wiener processes. We assume that $(\mathcal F_t)$ is the natural filtration of the processes $W,N$ and we assume also that $\displaystyle \mathcal F_\infty=\bigvee_{t>0}\mathcal F_t=\mathcal F$.
 
  On $(\Omega,\mathcal F,(\mathcal F_t),\mathbb P)$, we consider the following SDE
 \begin{equation}\label{eq:def_trajectory}
    \begin{split}
    \rho(t)=\rho_0 &+ \int_0^t L(\rho(s-))ds \\
		   &+ \sum_{i=0}^p\int_0^t H_i(\rho(s-)) dW_i(s) \\
		   &+ \sum_{i=p+1}^n \int_0^t \int_{\mathbb R} \left(\frac{J_i(\rho({s-}))}{v_i(\rho({s-}))}-\rho({s-})\right)\mathbf{1}_{0<x<v_i(\rho(s-))} [N_i(dx,ds)-dxds],
    \end{split}
   \end{equation}
 where $\rho_0\in\mathcal S(\mathcal H)$.
   
     \begin{definition}\label{def:trajectory}
   The equation \eqref{eq:def_trajectory} is called a stochastic master equation and its solution is called a quantum trajectory.
 \end{definition}
 
Equation \eqref{eq:def_trajectory} is a "generic"\footnote{One can generalize these equations by introducing time dependent and random coefficients \cite{Bar6}} SDE describing the evolution of a system undergoing continuous indirect measurements. Results of existence and uniqueness of the solution of \eqref{eq:def_trajectory} can be found in \cite{Pe1,Pe2,Pe3,Bar2,Bar6}.
 
 In Eq. \eqref{eq:def_trajectory}, the operator $L$ is a usual Lindblad operator \cite{Lin1,Gor1}. These operators appear in the definition of the master equation in the Markovian approach of Open Quantum Systems. 
 
From Eq. \eqref{eq:def_trajectory}, one can introduce the measurement record counting processes: $$\hat{N}_i(t)=\int_0^t \int_{\mathbb R}\mathbf{1}_{0<x<v_i(\rho(s-))} N_i(dx,ds), i=p+1,\ldots,n.$$
These processes are counting processes with stochastic intensity given by
 $$\int_0^t v_i(\rho(s-))ds,i=p+1,\ldots,n.$$
 In particular, for any $i\in\{p+1,\ldots, n\}$, the process $(\hat{N}_i(t)-\int_0^t v_i(\rho(s-))ds)$ is a $(\mathcal{F}_t)$ martingale under the probability $\mathbb P$. 

During an experiment, these processes would correspond to the counting measurement records an experimenter would obtain. For example $\hat{N}_i(t)$ could correspond to the total number of photons arrived on a detector up to time $t$.  In section \ref{sec:trial_state} we discuss in more details what would be the equivalent for a continuous measurement record. 

In terms of $\hat{N}_i(t)$, Eq. \eqref{eq:def_trajectory} can be written as
 \begin{eqnarray}\label{eq:def_trajectory2}
   d\rho(t)&=&L(\rho(t-))dt + \sum_{i=0}^p H_i(\rho(t-)) dW_i(t)\nonumber\\&& + \sum_{i=p+1}^n \left(\frac{J_i(\rho({t-}))}{v_i(\rho({t-}))}-\rho({t-})\right)(d\hat{N}_i(t)-v_i(\rho(t-))dt).
\end{eqnarray}
  
In the next section we introduce a nondemolition condition on this evolution and study the large time behaviour of $(\rho(t))$.
  
  \subsection{Non demolition condition}
  
A measurement process is called nondemolition if one can find a basis $\mathcal P$ of $\mathcal H$ such that any element of $\mathcal P$ is unmodified by the measurement process. If, at a given time, a system is in one of the basis states, it will remain in it at any future time with probability one.

\begin{definition}
Let $\mathcal P$ be a basis of $\mathcal H$. A measurement process fulfills a nondemolition condition for $\mathcal P$ if any state of $\mathcal P$ is stable under the measurement process: for any $\vert \alpha \rangle \in \mathcal P$ if at time $s$, $\rho(s)=\vert \alpha \rangle\langle\alpha\vert$ then for any time $t>s$, $\rho(t)=\rho(s)$, almost surely.
\end{definition}
The stable states $\vert \alpha\rangle\langle\alpha\vert, \alpha\in\mathcal P$, are called \textit{pointer states}.

We assume from now on that $H$ and the $C_i$'s are diagonal in the basis $\mathcal P$. The main result we prove in this section is the equivalence between this diagonal assumption and a nondemolition condition for $\mathcal P$.

 The diagonal assumption can be expressed as follows. There exist $\epsilon(\alpha) \in \mathbb R$, $c(i\vert\alpha)\in \mathbb C$ such that
  \begin{align*}
   H&=\sum_{ \alpha \in \mathcal P} \epsilon(\alpha) \vert \alpha\rangle\langle\alpha\vert\\
   C_i&=\sum_{ \alpha \in \mathcal P} c(i\vert \alpha) \vert \alpha\rangle\langle\alpha\vert \;,i=0,\ldots,n.
  \end{align*}
  Attached to these decompositions we introduce the following quantities which will be used further
   \begin{align*}
   r(i\vert\alpha)&=c(i|\alpha)+\overline{c(i|\alpha)},\,\, i=0,\ldots,n,\\
   \theta(i\vert\alpha)&=|c(i|\alpha)|^2,\,\, i=0,\ldots,n.
   \end{align*}
   Here $\overline{z}$ is $z$ complex conjugate. In the basis $\mathcal P$, we denote a matrix $A=(A_{\alpha\beta})_{\alpha,\beta}$.

Our study of $(\rho(t))$ is mainly based on the study of the diagonal elements of $(\rho(t))$ in the basis $\mathcal P$. If a direct measurement identifying all the pointer states would have been performed at time $t=0$, then the system after this measurement would have been in the pointer state $\vert \alpha\rangle\langle \alpha\vert$ with probability $\rho_{\alpha\alpha}(0)={\rm Tr}[\rho(0)\vert \alpha\rangle\langle \alpha\vert]$. If the same direct measurement is performed at time $t>0$ the probability to obtain the same system pointer state is $\rho_{\alpha\alpha}(t)={\rm Tr}[\rho(t)\vert \alpha\rangle\langle \alpha\vert]$.  So, the evolution of the diagonal elements of $(\rho(t))$ in $\mathcal P$ gives us information on the distribution of such direct measurement outcomes.

In the sequel, for all $\alpha\in\mathcal P$, we use the notations
\begin{equation}
q_\alpha(t)=\rho_{\alpha\alpha}(t)={\rm Tr}[\rho(t)\vert \alpha\rangle\langle \alpha\vert],\,\,t\geq0.
\end{equation}
As a preliminary, let us prove that the $(q_\alpha(t))$, $\alpha \in \mathcal P$ are $(\mathcal F_t)$ martingales solutions of Dade-Doleans type of SDEs.
 
 \begin{thm} On $(\Omega,\mathcal F,(\mathcal F_t),\mathbb P)$, the stochastic processes $(q_\alpha(t)),\alpha\in\mathcal P$ satisfy the following system of stochastic differential equations
 \begin{multline}\label{eq:q_alpha(t)}
     dq_\alpha(t)= q_\alpha(t-)\Bigg[\sum_{i=0}^p\big(r(i\vert \alpha) - \langle r_i(t-)\rangle\big) dW_i(t)
			    \hfill \\\hfill+ \sum_{i=p+1}^n \left(\frac{\theta(i\vert \alpha)}{\langle\theta_i(t-)\rangle} -1 \right) \Big(d\hat N_i(t)-\langle\theta_i(t-)\rangle dt\Big)\Bigg],\,\,\alpha\in\mathcal P,
   \end{multline} 
   where $\displaystyle\langle r_i(t)\rangle=\sum_{\gamma}r(i\vert \gamma)q_\gamma(t)$ and $\displaystyle\langle\theta_i(t)\rangle=\sum_{\gamma}\theta(i\vert \gamma)q_\gamma(t)$, for all $t\geq0$.
   
   In particular, the stochastic processes $(q_\alpha(t)),\alpha\in\mathcal P$ are $(\mathcal F_t)$ martingales. As solution of Dade-Doleans type of SDEs, they can be expressed in the following form
 \begin{multline}\label{eq_dade}
q_\alpha(t)=q_\alpha(0)\times\exp\Bigg[\sum_{i=0}^p\left(\int_0^t\big(r(i\vert \alpha)-\langle r_i(s-)\rangle\big)dW_i(s)-\frac{1}{2}\int_0^t\big(r(i\vert \alpha)-\langle r_i(s-)\rangle\big)^2ds\right)\Bigg]\hfill\\\hfill\times
\prod_{i=p+1}^n\prod_{s\leq t}\left(1+\left(\frac{\theta(i\vert\alpha)}{\langle\theta_j(s-)\rangle}-1\right)\Delta\hat{N}_i(t)\right)\times\exp \Bigg[-\int_0^t\left(\theta(i\vert\alpha)-\langle\theta_j(s-)\rangle\right))ds\Bigg]\\
\hphantom{q_\alpha(t)=}=q_\alpha(0)\times\exp\Bigg[\sum_{i=0}^p\left(\int_0^t\big(r(i\vert \alpha)-\langle r_i(s-)\rangle\big)dW_i(s)-\frac{1}{2}\int_0^t\big(r(i\vert \alpha)-\langle r_i(s-)\rangle\big)^2ds\right)\hfill\\\hfill
 +\sum_{i=p+1}^n\left(\int_0^t\ln\left(\frac{\theta(i\vert\alpha)}{\langle\theta_j(s-)\rangle}\right)d\tilde N_i(s)-\int_0^t\big(\theta(i\vert\alpha)-\langle\theta_j(s-)\rangle\big)ds\right)\Bigg].
\end{multline}  
   \end{thm}
In the case where $\theta(i|\alpha)=0$, for some $i=p+1,\ldots,n$, if a jump of the corresponding $\hat N_i$ occurs at some time $t$, one can see that $q_\alpha(t)$ vanishes ($q_\alpha(u)=0,$ for all $u\geq t$). In this situation in order to give a sense to the second expression one can consider that $\ln(0)=-\infty$ and $\exp(-\infty)=0$. Nevertheless the second expression will be used only when $\theta(i|\alpha)>0$ for all $\alpha\in\mathcal P$ and for all $i=p+1,\ldots,n$. Let us stress that in this situation, if $q_\alpha(0)\neq0$, we have $q_\alpha(t)>0,$ for all $t\geq0$. Although, in section \ref{sec:expo_conv}, we discuss some interesting properties of $(q_\alpha(t))$  when $\theta(i|\alpha)=0$ for some $i$.

  \begin{proof} In order to obtain the expression \eqref{eq:q_alpha(t)}, we have to compute $d\rho_{\alpha\alpha}(t)$ by using \eqref{eq:def_trajectory2}. To this end we have to plug the diagonal condition into the expression of $L$, $H_i$, $J_i$ and $v_i$. This way we can compute the following expression.
  \begin{equation}
   \begin{split}
    v_i(\rho)&=\sum_{\alpha \in \mathcal P} \vert c(i\vert \alpha)\vert^2\rho_{\alpha\alpha},\,\,i=p+1,\ldots,n\\
    (J_i(\rho))_{\alpha\beta}&=\rho_{\alpha\beta}\times c(i\vert \alpha)\overline{c(i\vert \beta)},\,\,i=p+1,\ldots,n\\
    (H_i(\rho))_{\alpha\beta}&=\rho_{\alpha\beta}\times\left(c(i\vert \alpha)+\overline{c(i\vert \beta)} - \sum_{\gamma\in\mathcal P}(c(i\vert \gamma)+\overline{c(i\vert \gamma)})\rho_{\gamma\gamma}\right),\;i=0,\ldots,p\\
    (L(\rho))_{\alpha\beta}&=\rho_{\alpha\beta}\times\left(-i(\epsilon_\alpha -\epsilon_\beta) + \sum_{i=0}^n c(i\vert \alpha)\overline{c(i\vert \beta)} - \frac12 (\vert c(i\vert \alpha) \vert^2  + \vert c(i\vert \beta) \vert^2\right).
   \end{split}
  \end{equation}
  
  From now on we only need the expression of the diagonal elements $\alpha=\beta$ in the pointer basis. In other words the stochastic differential equations for $ (q_\alpha(t))$ do not depend on the off diagonal elements of the system state. This way remarking also that $(L(\rho))_{\alpha\alpha}=0$ for any $\alpha \in \mathcal P$, it is easy to derive Equation \eqref{eq:q_alpha(t)}.
  
  The second part follows from the fact that the processes $(W_i(t))$ and $(\hat{N}_i(t)-\int_0^tv_i(\rho(s-))ds)$ are $(\mathcal F_t)$ martingales. This way the stochastic processes $(q_\alpha(t))$ are local martingales but since they are bounded they are true martingales. The solution \eqref{eq_dade} is the usual expression of the solution of a Dade-Doleans SDE. 
  \end{proof}
 
We are now equipped to prove the equivalence between the diagonal assumption and the nondemolition condition.
  \begin{thm}
   The quantum stochastic master equation \eqref{eq:def_trajectory2} fulfills a nondemolition condition for $\mathcal P$ if and only if $H$ and all the operators $C_i$ are diagonal in the basis $\mathcal P$. 
  \end{thm}

\begin{proof}
 Let us first prove that the diagonal condition imply the nondemolition condition. We need to prove that if at time $s$, $\rho(s)=\vert\alpha\rangle\langle\alpha\vert$, then at any time $t>s$, $\rho(t)=\vert \alpha\rangle\langle\alpha \vert$ almost surely. Since $\rho(t)$ is a Markovian process, we can, without loss of generality, set $s=0$.

Put $\rho(0)=\vert\alpha\rangle\langle\alpha\vert$ for some $\alpha\in\mathcal P$. We have $q_\beta(0)=0$ for any $\beta\neq\alpha, \beta\in \mathcal P$ and ${\rm Tr}[\rho(t)]=1$ for any time $t$. Then, looking at \eqref{eq_dade}, we have $q_\beta(t)=0$ and $q_\alpha(t)=1$ for all $t\geq0$ almost surely.  It implies
\[
\rho(t)=\vert \alpha\rangle\langle\alpha \vert, \,\, \forall t\geq0,\,\,\textrm{a.s}.
\]

We now prove that the nondemolition condition implies the diagonal assumption.  If at time $s$, $\rho(s)=\vert \alpha\rangle\langle\alpha\vert$, then for any time $t>s$, $\rho(t)=\vert \alpha\rangle\langle\alpha\vert,$ almost surely. The expectation of $\rho(t)$ conditioned on $\rho(s)=\vert \alpha\rangle\langle \alpha \vert$ with $\alpha \in \mathcal P$ is:
\begin{align*}
\mathbb E\Big[\rho(t)&\Big|\rho(s)=\vert \alpha\rangle\langle \alpha \vert\Big]-\rho(s)=\mathbb E\left[\left.  \int_s^t L(\rho(u-))du\right|\rho(s)=\vert \alpha\rangle\langle \alpha \vert\right] \\
&+\mathbb E\left[\left.  \sum_{i=0}^p \int_s^t H_i(\rho(u-)) dW_i(u)\right|\rho(s)=\vert \alpha\rangle\langle \alpha \vert\right] \\
&+\mathbb E\left[\left.  \sum_{i=p+1}^n\int_s^t\left(\frac{J_i(\rho({u-}))}{v_i(\rho({u-}))}-\rho({u-})\right)(d\hat{N}_i(u)-v_i(\rho(u-))du)\right|\rho(s)=\vert \alpha\rangle\langle \alpha \vert\right] .
\end{align*}
Since $(W_i(t))$ and $(\hat{N}_i(t)-\int_0^t v_i(\rho(u-))du)$ are  martingales,
\begin{align*}
\mathbb E\Big[\rho(t)\Big|\rho(s)=\vert \alpha\rangle\langle \alpha \vert\Big]-\rho(s)=\int_s^t\mathbb E\Big[   L(\rho(u))\Big|\rho(s)=\vert \alpha\rangle\langle \alpha \vert\Big]du.
\end{align*}
At this stage, since $\rho(u)=\rho(s)$ for all $u\geq s$ almost surely, we get
\begin{align*}
0=&\mathbb E\Big[\rho(t)\Big |\rho(s)=\vert \alpha\rangle\langle \alpha \vert\Big]-\rho(s)\\
=&L(\rho(s))(t-s),\forall t\geq s\,\,\textrm{a.s}.
\end{align*}

Let  $\beta \in \mathcal P, \beta \neq \alpha$. The condition $L(\vert \alpha\rangle\langle \alpha \vert)_{\beta\beta}=0$ implies $(C_i)_{\beta\alpha}=0$ for all $i=0,\ldots,n$.
Using this result, the condition $L(\vert \alpha\rangle\langle \alpha \vert)_{\beta\alpha}=0$ implies $H_{\beta\alpha}=0$. Hence, $H$ and all the $C_i$'s must be diagonal in the basis $\mathcal P$.
  \end{proof}

In the next section  we use the martingale property of $(q_\alpha(t))$ to study its long time behavior. Before let us prove that this martingale property is equivalent to the nondemolition condition.
  \begin{prop}
   The processes $(q_\alpha(t))$ are $(\mathcal F_t)$ martingales if and only if  \eqref{eq:def_trajectory2} fulfills a nondemolition condition for $\mathcal P$.
  \end{prop}
  \begin{proof}
   We already proved that the nondemolition condition implies the martingale property of $(q_\alpha(t))$. Let us prove the converse.
   
   We suppose that for any $\alpha \in \mathcal P$, the process $(\rho_{\alpha\alpha}(t))$ is a martingale.
This assumption is true only if the drift  part of \eqref{eq:def_trajectory2} for $\rho_{\alpha\alpha}(t)$ is null whatever is the initial state. Hence, for any arbitrary $\rho \in \mathcal S(\mathcal H)$, we must have  $L(\rho)_{\alpha\alpha}=0$. Take $\rho=\vert \beta\rangle\langle \beta \vert$ with $\alpha\neq\beta$. As seen earlier, the condition $L(\rho)_{\alpha\alpha}=0$, implies that the $C_i$'s must be diagonal in the basis $\mathcal P$. Now put $\vert \alpha + \beta \rangle =\frac1{\sqrt{2}}(\vert \alpha \rangle + \vert \beta \rangle)$. Take $\rho=\vert \alpha +\beta \rangle\langle \alpha +\beta \vert$. The condition $L(\rho)_{\alpha\alpha}=0$ implies $H_{\alpha\beta}-H_{\beta\alpha}=0$. Put $\vert \alpha +i\beta \rangle=\frac1{\sqrt{2}}(\vert \alpha \rangle + i\vert \beta \rangle)$. Take $\rho=\vert \alpha +i\beta \rangle\langle \alpha +i\beta \vert$. The condition $L(\rho)_{\alpha\alpha}=0$ implies $H_{\alpha\beta}+H_{\beta\alpha}=0$. We can thus conclude that $H$ must also be diagonal in the 
basis 
$\mathcal P$. 
   
   As proved earlier this diagonal property is equivalent to the nondemolition condition.
  \end{proof}
  
The next section is devoted to the large time behavior of $(q_\alpha(t))$ and to interpretations of the obtained convergence in terms of wave function collapse.
  
 \section{Convergence and wave function collapse}
 
 \subsection{Wave function collapse}
 In this section we show the almost sure convergence of the processes $(q_\alpha(t))$ when $t$ goes to infinity. Under some non degeneracy conditions, we can identify the limit random variables $(q_\alpha(\infty))$. More precisely, in this context we show that $q_\alpha(\infty)$ is equal to $1$ for a pointer $\Upsilon\in \mathcal P$ and $0$ for the others. The pointer state $\Upsilon$ is a random variable and we find its distribution. We next show that this imply that $(\rho(t))$ converges almost surely to one of the pointer states $\vert\alpha\rangle\langle\alpha\vert$. In particular, we show that the probability for the limit pointer state to be $\vert \alpha\rangle\langle\alpha\vert$ is $q_\alpha(0)=\rho_{\alpha\alpha}(0)$. This is what is predicted by the von Neumann projection postulate if a direct measurement would have been performed at time $0$. Thinking of the limit state in terms of random variable, in the limit $t\to \infty$, the system state is a random variable with the same law as the one 
predicted by von 
Neumann projection postulate for a 
direct measurement at time $0$.
 
In the following subsections we present some useful properties implied by this convergence.
 
 Let us express our non degeneracy condition
 
 \noindent\textbf{Assumption (ND):} For any $(\alpha,\beta)$ with $\alpha\neq\beta$ there exists $i \in \{0,1, \ldots, n\}$ such that
 \begin{itemize}
 \item either $r(i|\alpha)\neq r(i|\beta)$ if $i\le p$
 \item or $\theta(i|\alpha)\neq \theta(i|\beta)$ if $i>p.$
 \end{itemize}

 \begin{thm}\label{thm:conv}  Under Assumption (ND), there exist random variables $q_\alpha(\infty),\alpha\in\mathcal P$ which take values in $\{0,1\}$ such that 
 \begin{eqnarray}
 \lim_{t\rightarrow\infty}q_\alpha(t)&=&q_\alpha(\infty),\,\,\forall \alpha\in\mathcal P,
 \end{eqnarray}
almost surely and in $L^1$ norm. Moreover, we have $\mathbb E[q_\alpha(\infty)\vert\mathcal F_t]=q_\alpha(t)$, for all $\alpha\in\mathcal P$ and for all $t\geq 0$. 

The random variables $q_\alpha(\infty),\alpha\in\mathcal P$, satisfy
 \begin{eqnarray}
 \mathbb P(q_\alpha(\infty)=1)&=&q_0(\alpha), \forall\alpha\in\mathcal P,\\
 q_\alpha(\infty)q_\beta(\infty)&=&0,\,\, \forall\alpha\neq\beta,\,\, \textrm{a.s}.
 \end{eqnarray}

  As a consequence there exists a random variable $\Upsilon$ with values in $\mathcal P$ such that $$\mathbb P(\Upsilon=\alpha)=q_\alpha(0),\;\forall\alpha\in\mathcal P$$ and such that 
  $$\lim_{t \to \infty} \rho(t)=\vert\Upsilon\rangle\langle\Upsilon\vert,\,\,\textrm{a.s}.$$
 \end{thm}
 \begin{proof} Let $\alpha\in\mathcal P$ be fixed. The almost sure convergence of $(q_\alpha(t))$ follows from the fact that  $(q_\alpha(t))$ are bounded $(\mathcal F_t)$ martingales. More precisely, the family $(q_\alpha(t))$ is uniformly integrable. Therefore there exists $q_\alpha(\infty)$ such that $\lim_{t\rightarrow\infty}q_\alpha(t)=q_\alpha(\infty),$ almost surely and in $L^1$ norm and we have $\mathbb E[q_\alpha(\infty)\vert\mathcal F_t]=q_\alpha(t)$, for all $t\geq 0$. It remains to show that these random variables take values in $\{0,1\}$. To this end, for $\alpha$ being fixed, using It\^o-L\'evy isometry, we have for all $t>0$
 $$\mathbb E[q_\alpha(t)^2]=\mathbb E[[q_\alpha(t), q_\alpha(t)]],$$
 where $[q_\alpha(t), q_\alpha(t)]$ corresponds to the stochastic bracket of $(q_\alpha(t))$.
 We then have
 \begin{eqnarray}
\mathbb E[q_\alpha(t)^2]&=&\sum_{i=0}^p\int_0^t\mathbb E\left[q_\alpha(s)^2(r(i\vert \alpha)-\langle r_i(s)\rangle)^2\right]ds\nonumber\\
&&+\sum_{i=p+1}^n\int_0^t\mathbb E\left[q_\alpha(s)^2\left(\frac{\theta(i\vert \alpha)}{\langle\theta_i(s)\rangle} -1 \right)^2\langle\theta_i(s)\rangle\right] ds
 \end{eqnarray}
 Since the processes $(q_\alpha(t))$ converge almost surely and are bounded, using Lebesgue dominated convergence Theorem, we have that the quantity $\mathbb E[q_\alpha(t)^2]$ converges when $t$ goes to infinity. This implies that
 \begin{eqnarray}\label{cconv}
 \int_0^\infty\mathbb E\left[q_\alpha(s)^2(r(i\vert \alpha)-\langle r_i(s)\rangle)^2\right]ds&<&\infty,\,\, i=0,\ldots,p,\nonumber\\
 \int_0^\infty\mathbb E\left[q_\alpha(s)^2\left(\frac{\theta(i\vert \alpha)}{\langle\theta_i(s)\rangle} -1 \right)^2\langle\theta_i(s)\rangle\right] ds&<&\infty,\,\, i=p+1,\ldots,n.
 \end{eqnarray}
Again by the dominated convergence Theorem, the quantities
 \begin{eqnarray}&&\mathbb E\left[q_\alpha(t)^2(r(i\vert \alpha)-\langle r_i(t)\rangle)^2\right],\,\,i=0,\ldots,p,\,\, \nonumber\\&&\mathbb E\left[\left(\frac{\theta(i\vert \alpha)q_\alpha(t)}{\langle\theta_i(t)\rangle} -q_\alpha(t) \right)^2\langle\theta_i(t)\rangle\right],\,\, i=p+1,\ldots,n
 \end{eqnarray} converge when $t$ goes to infinity. Then from \eqref{cconv} it follows that necessarily
  \begin{eqnarray}
\lim_{t\to\infty} \mathbb E[q_\alpha(t)^2(r(i\vert \alpha)-\langle r_i(t)\rangle)^2]&=&0,\,\, i=0,\ldots,p\\
\lim_{t\to\infty} \mathbb E\left[q_\alpha(t)^2\left(\frac{\theta(i\vert \alpha)}{\langle\theta_i(t)\rangle} -1 \right)^2\langle\theta_i(t)\rangle\right]&=&0,\,\, i=p+1,\ldots,n.
 \end{eqnarray}
 The above $L^2$ convergences imply the almost sure convergences up to an extraction. More precisely, there exist subsequences $(t^i_n)$ such that for all $i=0,\ldots,n,\lim_{n\to\infty}t^i_n=\infty$ and such that almost surely
   \begin{eqnarray}
\lim_{n\to\infty} q_\alpha(t^i_n)^2(r(i\vert \alpha)-\langle r_i(t^i_n)\rangle)^2&=&0,\,\, i=0\ldots,p\\
\lim_{n\to\infty}q_\alpha(t^i_n)^2\left(\frac{\theta(i\vert \alpha)}{\langle\theta_i(t^i_n)\rangle} -1 \right)^2\langle\theta_i(t^i_n)\rangle &=&0,\,\, i=p+1,\ldots,n.
 \end{eqnarray}
Since the processes $(q_\alpha(t))$ converge almost surely, by uniqueness of the almost sure limits and using the boundedness of $\langle \theta_i(t)\rangle$, we can conclude that almost surely, for all $\alpha\in\mathcal P$
 \begin{eqnarray}
\lim_{t\to\infty} q_\alpha(t)^2(r(i\vert \alpha)-\langle r_i(t)\rangle)^2&=&0,\,\, i=0\ldots,p\\
\lim_{t\to\infty}q_\alpha(t)^2(\theta(i\vert \alpha) - \langle\theta_i(t)\rangle)^2 &=&0,\,\, i=p+1,\ldots,n.
 \end{eqnarray}
Then it follows that, almost surely, for all $\alpha\in\mathcal P$
  \begin{eqnarray}
 q_\alpha(\infty)(r(i\vert \alpha)-\langle r_i(\infty)\rangle)&=&0,\,\, i=0\ldots,p\\
 q_\alpha(\infty)\left(\theta(i\vert \alpha) -\langle\theta_i(\infty)\rangle \right) &=&0\,\, i=p+1,\ldots,n.
 \end{eqnarray}
  This way, almost surely, for all $\alpha\neq\beta$
    \begin{eqnarray}
 q_\alpha(\infty)\big(r(i\vert \alpha)-\langle r_i(\infty)\rangle\big)&=&q_\beta(\infty)\big(r(i\vert \beta)-\langle r_i(\infty)\rangle\big)=0,\,\, i=0\ldots,p\\
q_\alpha(\infty)\big(\theta(i\vert \alpha) -\langle\theta_i(\infty)\rangle \big) &=& q_\beta(\infty)\big(\theta(i\vert \beta) -\langle\theta_i(\infty)\rangle \big)=0\,\, i=p+1,\ldots,n.
 \end{eqnarray}
 It follows that, almost surely, for all $\alpha\neq\beta$
 \begin{eqnarray}
q_\alpha(\infty) q_\beta(\infty)(r(i\vert \alpha)-r(i\vert \beta))&=&0,\,\, i=0\ldots,p\\
q_\alpha(\infty)q_\beta(\infty)(\theta(i|\alpha)-\theta(i|\beta))&=&0,\,\, i=p+1,\ldots,n.
 \end{eqnarray}
 Finally using Assumption (ND) one can conclude that  
   \[q_\alpha(\infty)q_\beta(\infty)=0,\,\,\forall\alpha\neq\beta\,\,\textrm{a.s}.\]
 This way, there exists a set $\Omega'$ such that $\mathbb P(\Omega')=1$ and such that for all $\omega\in\Omega'$ there exists a unique $\Upsilon\in\mathcal P$ such that $q_\Upsilon(\infty)(\omega)\neq 0$ and for all $\beta\neq\Upsilon, q_\beta(\infty)(\omega)= 0$. Moreover, since ${\rm Tr}[\rho(t)]=1$, we have that $\sum_{\alpha \in \mathcal P}q_\alpha(\infty)(\omega)=1$. Therefore $q_\Upsilon(\infty)(\omega)=1$. Then we have proved that for all $\alpha,q_\alpha(\infty)\in\{0,1\}$. Using now the martingale property, we have $\mathbb E[q_\alpha(\infty)]=q_\alpha(0)$, which implies that $\mathbb P(q_\alpha(\infty)=1)=q_\alpha(0)$ and the first part of the theorem is proved.

For the second part, let us come back to the definition of $\Upsilon$. This defines a random variable taking values in the set of pointer states $\mathcal P$ (for $\omega\in\Omega\setminus\Omega'$, we can put $\Upsilon(\omega)=\psi$, where $\psi\notin\mathcal P$, this will appear with probability 0). It is then clear that
$$\mathbb P(\Upsilon=\alpha)=\mathbb P(q_\alpha(\infty)=1)=q_\alpha(0),\;\forall\alpha\in\mathcal P.$$
One can note in particular that $q_\alpha(\infty)=\mathbb 1_{\Upsilon=\alpha}$. Now, we are in the position to conclude the proof. Indeed, since $q_\Upsilon(\infty)=1$, we get
$$\lim_{t\to\infty}\mathrm{Tr}[\rho(t)\vert\Upsilon\rangle\langle\Upsilon\vert]=1,\,\,\textrm{a.s}.$$
Now let $\omega\in\Omega'$ and let $\mu$ be a limit of a convergent subsequence $(\rho(t_n)(\omega))$, we have that $\mathrm{Tr}[\mu\vert\Upsilon(\omega)\rangle\langle\Upsilon(\omega)\vert]=1$ which implies that $\mu=\vert\Upsilon(\omega)\rangle\langle\Upsilon(\omega)\vert$. Therefore all convergent subsequences of $(\rho(t)(\omega))$ converge to $\vert\Upsilon(\omega)\rangle\langle\Upsilon(\omega)\vert$ which implies
$$\lim_{t \to \infty} \rho(t)=\vert\Upsilon\rangle\langle\Upsilon\vert,\,\,\textrm{a.s}.$$
 \end{proof}
 
 This result is crucial in the following. In particular it will allow us to use a Girsanov transformation "in infinite time horizon".
 
 \subsection{Exponential rate of convergence}\label{sec:expo_conv}
 
In this section, we study the convergence speed of the processes $(q_\alpha(t))$. In particular, we shall show an exponential convergence. To this end, we study the following quantities
$$\lim_{t\to\infty}\frac{1}{t}\ln\left(\frac{q_\alpha(t)}{q_\gamma(t)}\right),\,\,\alpha,\gamma\in\mathcal P.$$
Since $(q_t(\alpha))$ can vanish in the case where $\theta(i\vert\alpha)=0$ for some $i$, this quantity can be finite or infinite. Furthermore, we limit our study to pointer states $\alpha,\gamma \in \mathcal P$ such that $q_\alpha(0)\neq 0$ and $q_\gamma(0)\neq 0$. Remark, if $q_\beta(0)=0$ for some pointer state $\beta$ we have $q_\beta(t)=0$ for any time $t$.

 First, let us start by studying the case where $q_\alpha(t)>0$ and $q_\gamma(t)>0$, for all $t\geq0$, almost surely. As already discussed this is ensured when for any $i \in \{p+1,\ldots,n\}$, $\theta(i|\alpha)>0$ and $\theta(i|\gamma)>0$. In this case, using \eqref{eq:q_alpha(t)}, we have almost surely, for all $t\geq0$,
 \begin{multline}\label{eq_rapport}
 \frac{q_\alpha(t)}{q_\gamma(t)} =
\frac{q_\alpha(0)}{q_\gamma(0)}\times\exp\Bigg[\sum_{i=0}^p\Bigg(\int_0^t(r(i\vert \alpha)-r(i\vert \gamma))dW_i(s)\hfill\\
 \hfill-\frac 12\int_0^t\Big[(r(i\vert \alpha)-\langle r_i(s-)\rangle)^2-(r(i\vert \gamma)-\langle r_i(s-)\rangle)^2\Big]ds\Bigg)\\\hfill
 +\sum_{i=p+1}^n\left(\int_0^t\ln\left(\frac{\theta(i\vert\alpha)}{\theta(i\vert\gamma)}\right)d\hat N_i(s)-\int_0^t\left(\theta(i\vert\alpha)-\theta(i\vert\gamma)\right)ds\right)\Bigg].
 \end{multline}
 This can be rewritten as 
 \begin{multline}
  \frac{q_\alpha(t)}{q_\gamma(t)}
=
\frac{q_\alpha(0)}{q_\gamma(0)}\times\exp\Bigg[\sum_{i=0}^p\left(\int_0^t(r(i\vert \alpha)-r(i\vert \gamma))dX^\gamma_i(s)-\frac 12\int_0^t(r(i\vert \alpha)-r(i\vert \gamma))^2ds\right)\hfill\\\hfill
 +\sum_{i=p+1}^n\left(\int_0^t\ln\left(\frac{\theta(i\vert\alpha)}{\theta(i\vert\gamma)}\right)d\hat M^\gamma_i(s)-\int_0^t\left(\theta(i\vert\alpha)-\theta(i\vert\gamma)\right)-\ln\left(\frac{\theta(i\vert\alpha)}{\theta(i\vert\gamma)}\right)\theta(i\vert\gamma))ds\right)\Bigg],
 \end{multline}  
 for all $t\geq 0$, where
 \begin{eqnarray}\label{eq:def_X_M}
 X^\gamma_i(t)&=&W_j(t)-\int_0^t\Big[r(i\vert\gamma)-\langle r_i(s-)\rangle\Big]ds,\,\, i=0,\ldots,p\nonumber\\
\hat M^\gamma_i(t)&=&\hat{N}_i(t)-\int_0^t\theta(i\vert\gamma)ds=\hat{N}_i(t)-\theta(i\vert\gamma)t,\,\,i=p+1,\ldots,n
 \end{eqnarray}
for all $t\geq0$. Now we shall see that these processes are martingales under a suitable change of measure. To this end we consider the family of probability measures $(\mathbb Q^t_\gamma)$ defined by 
$$d\mathbb Q^t_\gamma(\omega)=\frac{q_\gamma(t)(\omega)}{q_\gamma(0)}d\mathbb P(\omega),\,\,t\geq0,$$
for all $\gamma \in \mathcal P$ such that $q_\gamma(0)\neq0$.

This family of probability measure is consistent. Moreover, since $\mathbb E[q_\gamma(\infty)\vert\mathcal F_t]=q_\gamma(t)$, for all $t\geq 0$, any element of this family can be extended to a unique probability measure $\mathbb Q_\gamma$ on $\mathcal F_\infty=\mathcal F$. In particular, we have
\begin{equation}
d\mathbb Q_\gamma(\omega)=\frac{q_\gamma(\infty)(\omega)}{q_\gamma(0)}d\mathbb P(\omega)
\end{equation}
 and in terms of filtration we get the following Radon Nykodim formula
\begin{equation}
\mathbb E\left.\left[\frac{d\mathbb Q_\gamma(\omega)}{d\mathbb P(\omega)}\right\vert\mathcal F_t\right]=\frac{q_\gamma(t)}{q_\gamma(0)}=\frac{d\mathbb Q^t_\gamma(\omega)}{d\mathbb P(\omega)},\,\,t\geq0.
\end{equation}
This way we can consider the quantity $\displaystyle\lim_{t\to\infty}\frac{1}{t}\ln\left(\frac{q_\alpha(t)}{q_\gamma(t)}\right)$ under $\mathbb Q_\gamma$. We shall need the following lemma which relies on Girsanov transformation.

\begin{lem}\label{lem_changeofmeasure} Let $\gamma\in\mathcal P$ such that $q_\gamma(0)\neq0$.
Under $\mathbb Q_\gamma$, the processes $(X^\gamma_j(t)),j=0,\ldots,p$ and $(\hat M^\gamma_j(t)),j=p+1,\ldots,n$ are $(\mathcal F_t)$ martingales. More precisely $(X^\gamma_j(t)),j=0,\ldots,p$ are standard $\mathbb Q_\gamma$ Brownian motions and $(\hat{N}_j(t)),j=p+1,\ldots,n$ are usual Poisson processes with deterministic intensities $\theta(i\vert\gamma)$.
\end{lem}

The following theorem expresses the exponential convergence speed towards $\Upsilon$. 

\begin{thm}\label{thm:conv_expo} Assume Assumption (ND) is satisfied. Assume that $\alpha,\gamma \in \mathcal P$ are such that $q_\alpha(0)\neq0$, $q_\gamma(0)\neq0$, and such that $\theta(i\vert\alpha)>0$, $\theta(i\vert\gamma)>0$ for all $i=p+1,\ldots,n.$ Then, we have
\begin{multline}
 \lim_{t\to\infty}\frac{1}{t}\ln\left(\frac{q_\alpha(t)}{q_\gamma(t)}\right)=-\frac{1}{2}\sum_{i=0}^p(r(i\vert \alpha)-r(i\vert \gamma))^2+\sum_{i=p+1}^n\theta(i\vert\gamma)\left[1-\frac{\theta(i\vert\alpha)}{\theta(i\vert\gamma)}+\ln\left(\frac{\theta(i\vert\alpha)}{\theta(i\vert\gamma)}\right)\right],
 \end{multline}
 $\mathbb Q_\gamma$ almost surely.
 
 More generally, in terms of the random variable $\Upsilon$, we have 
\begin{multline}
 \lim_{t\to\infty}\frac{1}{t}\ln\left(\frac{q_\alpha(t)}{q_\Upsilon(t)}\right)=-\frac{1}{2}\sum_{i=0}^p(r(i\vert \alpha)-r(i\vert \Upsilon))^2+\sum_{i=p+1}^n\theta(i\vert\Upsilon)\left[1-\frac{\theta(i\vert\alpha)}{\theta(i\vert\Upsilon)}+\ln\left(\frac{\theta(i\vert\alpha)}{\theta(i\vert\Upsilon)}\right)\right],
 \end{multline}
 $\mathbb P$ almost surely.
\end{thm}

\begin{proof}

From \eqref{eq_rapport}, we have
 \begin{eqnarray}
 \frac{1}{t}\ln\left(\frac{q_\alpha(t)}{q_\gamma(t)}\right)&=&\frac{1}{t}\ln\left(\frac{q_\alpha(0)}{q_\gamma(0)}\right)\nonumber\\&&+\sum_{i=0}^p(r(i\vert \alpha)-r(i\vert \gamma))\frac{X^\gamma_j(t)}{t}-\frac{1}{2}(r(i\vert \alpha)-r(i\vert \gamma))^2\nonumber\\&&+\sum_{i=p+1}^n\ln\left(\frac{\theta(i\vert\alpha)}{\theta(i\vert\gamma)}\right)\frac{\hat M^\gamma_i(t)}{t}
 +\left(\theta(i\vert\gamma)-\theta(i\vert\alpha)\right)+\ln\left(\frac{\theta(i\vert\alpha)}{\theta(i\vert\gamma)}\right)\theta(i\vert\gamma).\nonumber\\
 \end{eqnarray}
 Now from Lemma \ref{lem_changeofmeasure}, we know that under $\mathbb Q_\gamma$, the processes $(X^\gamma_j(t))$ are standard Brownian motions and $(\hat{N}_i(t))$ are a usual Poisson processes of intensities $\theta(i\vert\gamma)$. This way, using law of large number for Brownian motion and Poisson processes, we get
\begin{multline}
 \lim_{t\to\infty}\frac{1}{t}\ln\left(\frac{q_\alpha(t)}{q_\gamma(t)}\right)=-\frac{1}{2}\sum_{i=0}^p(r(i\vert \alpha)-r(i\vert \gamma))^2+\sum_{i=p+1}^n\theta(i\vert\gamma)\left[1-\frac{\theta(i\vert\alpha)}{\theta(i\vert\gamma)}+\ln\left(\frac{\theta(i\vert\alpha)}{\theta(i\vert\gamma)}\right)\right],
 \end{multline}
 $\mathbb Q_\gamma$ almost surely.
 
  Now let us remark that 
 $$d\mathbb P(\omega)=\sum_{\gamma\; s.t.\; q_\gamma(0)>0} q_\gamma(0)d\mathbb P(\omega\vert\Upsilon(\omega)=\gamma)$$
 which yields
 $$d\mathbb Q_\gamma(\omega)=d\mathbb P(\omega\vert\Upsilon(\omega)=\gamma).$$
 As a consequence 
 \begin{multline}
 \lim_{t\to\infty}\frac{1}{t}\ln\left(\frac{q_\alpha(t)}{q_\Upsilon(t)}\right)=-\frac{1}{2}\sum_{i=0}^p(r(i\vert \alpha)-r(i\vert \Upsilon))^2+\sum_{i=p+1}^n\theta(i\vert\Upsilon)\left[1-\frac{\theta(i\vert\alpha)}{\theta(i\vert\Upsilon)}+\ln\left(\frac{\theta(i\vert\alpha)}{\theta(i\vert\Upsilon)}\right)\right],
 \end{multline}
 $\mathbb P$ almost surely and the result follows.
 \end{proof}
 
 From the above Theorem since $q_\Upsilon(\infty)=1$, $\mathbb P$ almost surely, we get the following asymptotic expansion. For all $\alpha\in\mathcal P$ such that $q_\alpha(0)\neq0$
  \begin{multline}
 q_\alpha(t)=e^{-t\Big[\frac{1}{2}\sum_{i=0}^p(r(i\vert \alpha)-r(i\vert \Upsilon))^2-\sum_{i=p+1}^n\theta(i\vert\Upsilon)\left[1-\frac{\theta(i\vert\alpha)}{\theta(i\vert\Upsilon)}+\ln\left(\frac{\theta(i\vert\alpha)}{\theta(i\vert\Upsilon)}\right)\right]\Big]}\times(1+\circ(1)),\hfill
 \end{multline}
 $\mathbb P$ almost surely. From the inequality $\ln(x)\leq x-1$, we see that the rate
\begin{equation}\label{expo}
\frac{1}{2}\sum_{i=0}^p(r(i\vert \alpha)-r(i\vert \Upsilon))^2-\sum_{i=p+1}^n\theta(i\vert\Upsilon)\left[1-\frac{\theta(i\vert\alpha)}{\theta(i\vert\Upsilon)}+\ln\left(\frac{\theta(i\vert\alpha)}{\theta(i\vert\Upsilon)}\right)\right]\geq 0,\,\,\textrm{a.s}.\end{equation}
  More precisely each term of the sums is nonnegative. Now from Assumption (ND), the quantity \eqref{expo} is equal to zero if and only if $\alpha=\Upsilon$. This underlines the exponential rate convergence towards $\Upsilon$.

\begin{remark}
Let us assume the interaction of our measurement apparatus with the system involves only one hermitian operator $C_i=C$. In other words $n=0$. Let us also assume we can choose either to have a continuous process as our measurement record or a counting process. In other words, either the quantum stochastic master equation is
\begin{equation}\label{eq:only diff}
 d\rho(t)=L(\rho(t))dt + \Big(C\rho(t) + \rho(t)C^* - {\rm Tr}[(C+C^*)\rho(t)]\rho(t)\Big)dW_t
\end{equation}
or
\begin{equation}\label{eq:only jump}
 d\rho(t)=L(\rho(t-))dt + \left(\frac{C\rho(t-)C^*}{{\rm Tr}[C^*C\rho(t-)]} -\rho(t-)\right)[d\hat{N}(t) - {\rm Tr}[C^*C\rho(t-)]dt].
\end{equation}

In the diffusive case \eqref{eq:only diff}, the convergence rate \eqref{expo} is equal to $(c(\alpha)-c(\Upsilon))^2$. In the counting case \eqref{eq:only jump} it is equal to $-c(\Upsilon)^2[\ln(c(\alpha)^2/c(\Upsilon)^2) + 1 - c(\alpha)^2/c(\Upsilon)^2 ]$. A simple study shows that 
\[(c(\alpha)-c(\Upsilon))^2\leq-c(\Upsilon)^2[\ln(c(\alpha)^2/c(\Upsilon)^2) + 1 - c(\alpha)^2/c(\Upsilon)^2 ].
\]
So the choice of a counting process may lead to a higher convergence rate. But it comes at a price. Suppose $C$ has two different eigenvalues of equal norm: $c(\alpha)\neq c(\beta)$, $|c(\alpha)|=|c(\beta)|$. The non degeneracy assumption (ND) is not fulfilled for the jump equation \eqref{eq:only jump} whereas it is fulfilled for the diffusive equation \eqref{eq:only diff}.
\end{remark}

We now study the situation where it exists $i \in \{p+1,\ldots,n\}$ and $\alpha\in\mathcal P$, such that $\theta(i|\alpha)=0$. In this case we shall study the following stopping times
\begin{eqnarray}
T_i&=&\min \{t \in \mathbb R_+ \textrm{ s.t. } \hat{N}_i(t)>0\},i=p+1,\ldots,n\\
 T(\alpha)&=&\min\{T_i\textrm{ s.t. } \theta(i|\alpha)=0\},\alpha\in\mathcal P.
\end{eqnarray}
 Assume that $q_\alpha(0)\neq0$, we have $q_\alpha(T(\alpha))=0$ and $q_\alpha(t)>0$, for all $t< T(\alpha)$ as well as $q_\alpha(t)=0$, for all $t\geq T(\alpha)$. This way, if $T(\alpha)<\infty$, the process $(q_\alpha(t))$ converges to zero in finite time.

We have the following proposition which describes the distribution of $T(\alpha),\alpha\in\mathcal P$.

\begin{prop}
Let $\alpha\in\mathcal P$ such that there exists $i\in\{p+1,\ldots,n\}$ such that $\theta(i\vert\alpha)=0$. Then,
\[
\mathbb P(T(\alpha)\leq t|\Upsilon=\gamma)=1-e^{-\lambda(\alpha|\gamma)t},
\]
where $\displaystyle\lambda(\alpha|\gamma)=\sum_{i,\,\,s.t.\,\,\theta(i\vert\alpha)=0}\theta(i\vert\gamma)$.

 Finally, we have
\[
\mathbb P(T(\alpha)\leq t)=1-\sum_{\beta\in \mathcal P} q_\beta(0)e^{-\lambda(\alpha|\beta)t}.
\]

\end{prop}

\begin{proof}
Remember that, under $\mathbb Q_\gamma$, $(\hat{N}_i(t))$ are usual Poisson processes of intensities $\theta(i|\gamma)$. So, under $\mathbb Q_\gamma$, the law of $T_i$ is exponential with parameter equal to $\theta(i|\gamma)$, that is,
\[
\mathbb Q_\gamma(T_i\leq t)=1-e^{-\theta(i|\gamma)t}, i=p+1,\ldots,n.
\]
Under $\mathbb Q_\gamma$, $T(\alpha)$ is thus the minimum of a finite set of random variables obeying exponential laws. Since we have assumed that the Poisson point processes $N_i(.,.)$ are independent, under $\mathbb Q_\gamma$ the processes $(\hat{N}_i(t))$ are also independent as well as the stopping times $T_i$. From  the properties of exponential law, we have
\[
\mathbb Q_\gamma(T(\alpha)\leq t)=1-e^{-\lambda(\alpha|\gamma)t}.
\]
Since $\mathbb P=\sum_{\beta \in \mathcal P} q_\beta(0)\mathbb Q_\beta$, we have
\[
\mathbb P(T(\alpha)\leq t)=1-\sum_{\beta \in \mathcal P}q_\beta(0) e^{-\lambda(\alpha|\beta)t},
\]
which concludes the proof.
\end{proof}

Let us note that taking $t$ goes to infinity we get $$\displaystyle\mathbb P(T(\alpha)=\infty)=\sum_{\beta,\,\,\textrm{s.t.}\lambda(\alpha,\beta)=0}q_\beta(0)\geq q_\alpha(0)=\mathbb P(q_\alpha(\infty)=1).$$

In the next section, we address the problem of convergence when one does not have access to the process $(\rho(t))$ but only to the measurement records.

 \subsection{Stability}\label{sec:trial_state}
 
Usually, in experiments the initial state $\rho_0$ of the system $\mathcal H$ is unknown (this is sometimes that we want to estimate). This way we cannot have access to the quantum trajectory $(\rho(t))$. Nevertheless we have still access to the results given by the measurement apparatus. These results are directly connected to the quantum trajectory \eqref{eq:def_trajectory}. In terms of processes, the results of the measurement are described by output processes in the following way. The observed processes are given by 
 $$dy_i(t)=dW_i(t)+\textrm{Tr}[(C_i+C_i^*)\rho(t-)]dt, i=0,\ldots,p$$
 for the diffusive part of the evolution and by
 $$d\hat N_i(t),i=p+1,\ldots,n$$
 for the counting processes. The quantities $y_i(t)$ and $\hat N_i(t)$ are the measurements recorded by the apparatus. In an homodyne or heterodyne detection scheme, $y_i(t)$ would represent the detected photo current integrated up to time $t$ whereas, in a direct photodetection scheme, $\hat N_i(t)$ would be the number of photons detected up to time $t$\cite{carm0,W1}. The quantum trajectory can be expressed as follows
  \begin{eqnarray}\label{eq:def_tttrajectory2}
   \rho(t)&=& L(\rho(t-))dt + \sum_{i=0}^p H_i(\rho(t-)) (dy_i(t)-\textrm{Tr}[(C_i+C_i^*)\rho(t-)]dt)\nonumber\\&& + \sum_{i=p+1}^n \left(\frac{J_i(\rho({t-}))}{v_i(\rho({t-}))}-\rho({t-})\right)(d\hat{N}_i(t)-v_i(\rho(t-))dt),\\
   \rho(0)&=&\rho_0.
\end{eqnarray}
 In the case where we do not know the initial state $\rho_0$, we can estimate the quantum trajectory $(\rho(t))$ by using an estimate quantum trajectory $(\tilde\rho(t))$ satisfying the following stochastic differential equation
 \begin{eqnarray}\label{eq:def_estimatetrajectory2}
   \tilde\rho(t)&=& L(\tilde\rho(t-))dt + \sum_{i=0}^p H_i(\tilde\rho(t-)) (dy_i(t)-\textrm{Tr}[(C_i+C_i^*)\tilde\rho(t-)]dt)\nonumber\\&& + \sum_{i=p+1}^n \left(\frac{J_i(\tilde\rho({t-}))}{v_i(\tilde\rho({t-}))}-\tilde\rho({t-})\right)(d\hat{N}_i(t)-v_i(\tilde\rho(t-))dt),\\
   \tilde\rho(0)&=&\tilde\rho_0,
\end{eqnarray}
where $\tilde\rho_0$ is an arbitrary state. Let us stress that $(y_i(t)),i=0,\ldots,p$ and $(\hat{N}_i(t)),i=p+1,\ldots,n$ are the output processes attached to the true quantum trajectory $(\rho(t))$. This way, if $\tilde \rho_0=\rho_0$ we get $(\tilde\rho(t))=(\rho(t))$. In particular the estimate quantum trajectory $( \tilde\rho(t))$ is governed by the measurement record as if it was the true quantum trajectory. This can allow to simulate an estimation of the true quantum trajectory. Such an estimate is often called a \textit{quantum filter}.  An important question is to know if the estimate become closer and closer to the true quantum trajectory when $t$ goes to infinity. In particular does the distance between the estimate and the true quantum trajectory converges to zero? Such a property is related to the so called  \textit{stability of quantum filter}. For general schemes of indirect measurement, partial results in this direction has been 
developed in \cite{Rou3} (it is shown that the fidelity 
between $(\rho(t))$ and $(\tilde\rho(t))$ increases at least in average when $t$ increases). 

Here, in the context of QND we show that the estimate quantum trajectory converges to the same state than the one of the true quantum trajectory. We shall show namely that $(\tilde\rho(t))$ converges to $\vert\Upsilon\rangle\langle\Upsilon\vert$ when $t$ goes to infinity. This is achieved by a direct analysis of the quantities
$$\tilde q_\alpha(t)=\textrm{Tr}[\tilde\rho(t)\vert\alpha\rangle\langle\alpha\vert],$$
for all $\alpha\in\mathcal P$ and for all $t\geq0$. The processes $(\tilde q_\alpha(t))$ repesent actually the estimate of the true $(q_\alpha(t))$.

From Equation \eqref{eq:def_estimatetrajectory2}, we can see that the processes $(\tilde q_\alpha(t))$ are solutions of
 \begin{multline}\label{eq:tilde q_alpha(t)}
     d\tilde q_\alpha(t)= \tilde q_\alpha(t-)\Bigg[\sum_{i=0}^p\Big(r(i\vert \alpha) - \langle \tilde r_i(t-)\rangle\Big) \Big(dW_i(t)+\big(\langle r_i(t-)\rangle-\langle\tilde r_i(t-)\rangle\big)\Big)
			    \hfill \\\hfill+ \sum_{i=p+1}^n \left(\frac{\theta(i\vert \alpha)}{\langle\tilde\theta_i(t-)\rangle} -1 \right) \Big(d\hat N_i(t)-\langle\tilde\theta_i(t-)\rangle dt\Big)\Bigg].	    
   \end{multline}

 Let us treat the case where $\theta(i\vert\alpha)\neq0,$ for all $\alpha\in\mathcal P$ and for all $i=p+1,\ldots,n$. Eq. \eqref{eq:tilde q_alpha(t)} are still Dade-Doleans exponential and the solution of \eqref{eq:tilde q_alpha(t)} are given by

 \begin{multline}\label{eq_dade_trial}
\tilde q_t(\alpha)=\tilde q_0(\alpha)\times\exp\Bigg[\sum_{i=0}^p\Bigg(\int_0^t\Big(r(i\vert \alpha)-\langle \tilde r_i(s-)\rangle\Big)\Big(dW_i(s)+\big(\langle r_i(s-)\rangle-\langle\tilde r_i(s-)\rangle\big)ds\Big)\hfill\\\hfill-\frac 12\int_0^t\big(r(i\vert \alpha)-\langle \tilde r_i(s-)\rangle\big)^2ds\Bigg)\hfill\\\hfill
 +\sum_{i=p+1}^n\Bigg(\int_0^t\ln\left(\frac{\theta(i\vert\alpha)}{\langle\tilde\theta_i(s-)\rangle}\right)d\hat N_i(s)-\int_0^t\left(\theta(i\vert\alpha)-\langle\tilde\theta_i(s-)\rangle\right))ds\Bigg)\Bigg],
\end{multline} 
for all $\alpha\in\mathcal P$ and for all $t\geq0$.  

Let us stress that, in general, these processes are no more $(\mathcal F_t)$ martingales. We cannot then conclude to their convergence by using martingale convergence results. Nevertheless the quantity $\tilde q_t(\alpha)/\tilde q_t(\gamma)$ takes exactly the same form than the one of the true quantum trajectory. More precisely, assume that $\tilde q_\gamma(0)\neq 0$, after computations we get that
 \begin{multline}\label{eq_rapport_trial}
 \frac{\tilde q_\alpha(t)}{\tilde q_\gamma(t)}=
\frac{\tilde q_\alpha(0)}{\tilde q_\gamma(0)}\times\exp\Bigg[\sum_{i=0}^p\Bigg(\int_0^t\big(r(i\vert \alpha)-r(i\vert \gamma)\big)dX^\gamma_i(s)-\frac 12\int_0^t(r(i\vert \alpha)-r(i\vert \gamma))^2ds\Bigg)\hfill\\\hfill
 +\sum_{i=p+1}^n\Bigg(\int_0^t\ln\left(\frac{\theta(i\vert\alpha)}{\theta(i\vert\gamma)}\right)d\hat M^\gamma_i(s)-\int_0^t\left(\theta(i\vert\alpha)-\theta(i\vert\gamma)\right)-\ln\left(\frac{\theta(i\vert\alpha)}{\theta(i\vert\gamma)}\right)\theta(i\vert\gamma))ds\Bigg)\Bigg],
 \end{multline} 
for all $\alpha,\gamma\in\mathcal P$, where the processes $(X^\gamma_i(t)),i=0,\ldots,p$ and $(M^\gamma_i(t)),i=p+1,\ldots,n$ have been defined in \eqref{eq:def_X_M}. Now we would like to consider the limit of this quantity under $\mathbb Q_\gamma$. To this end, we first need to consider a $\gamma$ such that $q_\gamma(0)\neq 0$. In this case, if $\tilde q_\alpha(0)\neq0$, we get 
 \begin{multline}
 \lim_{t\to\infty}\frac{1}{t}\ln\left(\frac{\tilde q_\alpha(t)}{\tilde q_\gamma(t)}\right)=-\frac{1}{2}\sum_{i=0}^p(r(i\vert \alpha)-r(i\vert \gamma))^2+\sum_{i=p+1}^n\theta(i\vert\gamma)\left[1-\frac{\theta(i\vert\alpha)}{\theta(i\vert\gamma)}+\ln\left(\frac{\theta(i\vert\alpha)}{\theta(i\vert\gamma)}\right)\right],
 \end{multline}
 $\mathbb Q_\gamma$ almost surely. As a consequence, we get again
 \begin{multline}\label{zaz}
 \lim_{t\to\infty}\frac{1}{t}\ln\left(\frac{\tilde q_\alpha(t)}{\tilde q_\Upsilon(t)}\right)=-\frac{1}{2}\sum_{i=0}^p(r(i\vert \alpha)-r(i\vert \Upsilon))^2+\sum_{i=p+1}^n\theta(i\vert\Upsilon)\left[1-\frac{\theta(i\vert\alpha)}{\theta(i\vert\Upsilon)}+\ln\left(\frac{\theta(i\vert\alpha)}{\theta(i\vert\Upsilon)}\right)\right],
 \end{multline}
$\mathbb P$ almost surely. This yields that $\mathbb P$ almost surely
 \begin{multline}\label{azaz}
\tilde q_\alpha(t)=\tilde q_\Upsilon(t)e^{-t\Big[\frac{1}{2}\sum_{i=0}^p(r(i\vert \alpha)-r(i\vert \Upsilon))^2-\sum_{i=p+1}^n\theta(i\vert\Upsilon)\left[1-\frac{\theta(i\vert\alpha)}{\theta(i\vert\Upsilon)}+\ln\left(\frac{\theta(i\vert\alpha)}{\theta(i\vert\Upsilon)}\right)\right]+\circ(1)\Big]}.\hfill
 \end{multline}
Now we consider $\Omega''$ with $\mathbb P(\Omega'')=1$ and such that \eqref{azaz} is fulfilled for all $\omega\in\Omega''$. Let $\omega\in\Omega''$ be fixed. For any $\alpha$ such that $\alpha\neq\Upsilon(\omega)$, under Assumption (ND) there exists $i$ such that either $r(i\vert\alpha)\neq r(i\vert \Upsilon(\omega))$ or $\theta(i\vert\alpha)\neq\theta(i\vert\Upsilon(\omega))$. Hence, since $0\leq \tilde q_\Upsilon(t)\leq 1$, we can conclude that
$$\lim_{t\to\infty}\tilde q_\alpha(t)(\omega)=0.$$
Now, since $\sum \tilde q_\alpha(t)=1$ still holds, we deduce that 
$$\lim_{t\to\infty}\tilde q_{\Upsilon(\omega)}(t)(\omega)=1.$$

Note that the result \eqref{zaz} requires only that $\tilde q_\alpha(0)\neq0$ as soon as $q_\alpha(0)\neq 0$. Then, we can express the following result.
\begin{prop} Assume that $\theta(i\vert\alpha)>0$, for all $\alpha\in\mathcal P$ and for all $i=p+1,\ldots,n.$
Assume Assumption (ND) is satisfied. Assume that $\tilde{q}_\alpha(0)\neq 0$ for any $\alpha \in \mathcal P$ such that $q_\alpha(0)\neq0$.

 Let $(\tilde q_\alpha(t))$ be the stochastic processes defined by \eqref{eq_dade_trial} with $\tilde{q}_\alpha(0)\neq0$. We have 
\begin{eqnarray}\lim_{t\to\infty}\tilde q_\Upsilon(t)&=&1,\nonumber\\
\lim_{t\to\infty}\tilde q_\alpha(t)\mathbf 1_{\Upsilon\neq\alpha}&=&0,
\end{eqnarray}
 $\mathbb P$ almost surely. Moreover, the convergence is exponentially fast, that is
  \begin{multline}
 \lim_{t\to\infty}\frac{1}{t}\ln\left(\frac{\tilde q_\alpha(t)}{\tilde q_\Upsilon(t)}\right)=-\frac{1}{2}\sum_{i=0}^p(r(i\vert \alpha)-r(i\vert \Upsilon))^2+\sum_{i=p+1}^n\theta(i\vert\Upsilon)\left[1-\frac{\theta(i\vert\alpha)}{\theta(i\vert\Upsilon)}+\ln\left(\frac{\theta(i\vert\alpha)}{\theta(i\vert\Upsilon)}\right)\right],
 \end{multline}
  $\mathbb P$ almost surely.
  
 Finally
 $$\lim_{t\to\infty}\tilde{\rho}(t)=\vert\Upsilon\rangle\langle\Upsilon\vert,$$
 $\mathbb P$ almost surely.
\end{prop}
In particular, it appears clearly that $$\lim_{t\to\infty}\tilde\rho(t)=\lim_{t\to\infty}\rho(t)=\vert\Upsilon\rangle\langle\Upsilon\vert$$ which shows the stability of the estimation and the convergence rate is the same. 

Remark that if $\tilde q_\alpha(0)=0$, then the conditions of our proposition impose $q_\alpha(0)=0$. Otherwise the estimation would be irrelevant. In the case $q_\alpha(0)=0$, this does not necessarily impose that $\tilde q_\alpha(0)=0$. In the case $\tilde q_\alpha(0)\neq0$, we still have that $\lim_{t\to\infty}\tilde q_\alpha(t)=0,$ almost surely (exponentially fast). In particular, in experiment a safe choice is to choose $\tilde\rho(0)$ such that $\tilde q_\alpha(0)\neq0,$ for all $\alpha\in\mathcal P$.
\bigskip

In the situation where it exists $i \in \{p+1,\ldots,n\}$, such that $\theta(i|\alpha)=0$, it is interesting to note that when a $q_\alpha(t)$ vanishes, the same happens for the estimate $\tilde q_\alpha(t)$. This follows from the fact that we have access to the jumping times of the processes $\hat N_i(.)$ through the measurement records. 

\subsection*{Acknowledgments}
C.~P.~ acknowledges financial support from the ANR project HAM-MARK,
N${}^\circ$ ANR-09-BLAN-0098-01.  

\noindent T.~B.~ thanks Denis Bernard for helpful discussions and acknowledges financial support from ANR contracts ANR-2010-BLANC-0414.01 and ANR-2010-BLANC-0414.02.

\end{document}